\documentclass[12pt]{article}

\usepackage{bm}
\usepackage{amssymb}
\usepackage{amsmath}
\usepackage{amsthm}
\usepackage{dsfont}
\usepackage{thmtools} 
\usepackage{thm-restate}
\usepackage{booktabs}
\usepackage{multirow}
\usepackage{adjustbox}
\usepackage{xcolor}
\usepackage{natbib}
\usepackage{hyperref}
\usepackage{url}
\usepackage{appendix}
\usepackage[ruled,vlined]{algorithm2e}

\def\mathP{\mathrm{Pr}}

\def\mathR{\mathbb{R}}
\def\equdist{\stackrel{\text{d}}{=}}

\newtheorem{theorem}{Theorem}
\newtheorem{corollary}{Corollary}

\newtheorem{lemma}{Lemma}

\newcommand{\blind}{1}

\begin{document}

\def\spacingset#1{\renewcommand{\baselinestretch}%
{#1}\small\normalsize} \spacingset{1}

\if1\blind
{
	\title{Tractable Bayes of Skew-Elliptical Link Models for Correlated Binary Data}
	\author{Zhongwei Zhang\footnotemark[1]\thanks{CEMSE Division, King Abdullah University of Science and Technology, Thuwal, Saudi Arabia}, Reinaldo B. Arellano-Valle\footnotemark[2]\thanks{Departamento of Estad\'istica, Pontificia Universidad Cat\'olica de Chile, Santiago, Chile}, Marc G. Genton\footnotemark[1], Rapha\"el Huser\footnotemark[1]}
	\date{}
	\maketitle
} \fi

\if0\blind
{
	\bigskip
	\bigskip
	\bigskip
	\begin{center}
		{\LARGE\bf Title}
	\end{center}
	\medskip
} \fi

\medskip

\begin{abstract} 
Correlated binary response data with covariates are ubiquitous in longitudinal or spatial studies.
Among the existing statistical models the most well-known one for this type of data is the multivariate probit model, which uses a Gaussian link to model dependence at the latent level.
However, a symmetric link may not be appropriate if the data are highly imbalanced.
Here, we propose a multivariate skew-elliptical link model for correlated binary responses, which includes the multivariate probit model as a special case.
Furthermore, we perform Bayesian inference for this new model and prove that the regression coefficients have a closed-form unified skew-elliptical posterior.
The new methodology is illustrated by application to COVID-19 data from three different counties of the state of California, USA.
By jointly modeling extreme spikes in weekly new cases, our results show that the spatial dependence cannot be neglected.
Furthermore, the results also show that the skewed latent structure of our proposed model improves the flexibility of the multivariate probit model and provides better fit to our highly imbalanced dataset.
\end{abstract}
\noindent%

\spacingset{1.4} % DON'T change the spacing!
%\clearpage

%%%%%%%%%%%%%%%%%%%%%%%%%%%%%%%%%%%%%%%%%%%%%%%%%%%%%%%%%%%%%%%%%%%%%
\section{Introduction}\label{intro}
Correlated binary response data with covariates frequently arise in longitudinal \citep{Fitzmaurice1995,Fitzmaurice2008} or spatial studies \citep{Heagerty1998,LinClayton2005}.
For instance, in longitudinal studies, the disease status (i.e., diseased or not diseased) is measured over time on the same person.
Similarly, in a panel study of income dynamics, the employment status information may be collected over time from the same survey participant.
The multivariate probit model \citep{Ashford1970, Chib1998} is well-known for this type of data, as it describes the dependence between binary variables by a latent Gaussian link, which allows for flexible modelling of dependence, has straightforward interpretation of the parameters and is easily amenable to Bayesian inference.

A symmetric link, however, does not always provide the best fit to a given dataset; see \cite{Chen1999}, \cite{Kim2008} for some examples.
In this case, the link might be misspecified, yielding substantial bias in the mean response estimates \citep{Czado1992}.
\cite{Chen1999} used the rate at which the probability of a given binary response variable approaches 0 and 1 to guide the selection of a symmetric or asymmetric link.
In other words, if the binary response data are highly imbalanced, the rate of the probability of the random variable approaching 0 is typically very different from the one approaching 1, so that an asymmetric link might be preferred than a symmetric link.
Motivated by this observation, a variety of flexible asymmetric link models have been proposed for univariate binary response data.
Surprisingly, to the best of our knowledge, no multivariate asymmetric link models have previously been proposed in the literature for correlated binary responses.
The purpose of this paper is to fill this important gap by proposing a flexible multivariate skew-elliptical link model for correlated binary responses, which includes the multivariate probit model as a special case and allows for fast and accurate Bayesian inference; see Section 2.2 for details on the multivariate skew-elliptical distribution, used in our model as a key building block.

\cite{Durante2019} has proved that for the univariate probit model with Gaussian priors, the posterior of the regression coefficients belongs to the class of unified skew-normal distributions \citep{AreAzza2006}.
This result has led to a similar result for the multinomial probit model \citep{FasanoDurante2020} and a closed-form predictive probability in probit models with Gaussian process priors \citep{CaoDurante2020}.
In this paper we also consider Bayesian inference for our new multivariate model and prove that the posterior of the regression coefficients belongs to the unified skew-elliptical family.
The closed-form and tractable posterior for the regression coefficients facilitates inference by using an algorithm which does not rely on data-augmentation, and thus avoids the convergence and mixing issues of the classical data-augmentation algorithms for probit models; see \cite{Johndrow2019} for a discussion of this issue.

We illustrate the new methodology by application to COVID-19 pandemic data from three different counties of the state of California, USA.
By jointly modeling the occurrences of extreme spikes in weekly new infected cases using our new model, we can estimate the underlying spatial dependence structure, which might provide helpful quantitative insights into the transmission modes of the virus and help authorities mitigate its spread.
Furthermore, our model has additional skewness parameters compared to the multivariate probit model, which improves its flexibility and makes it more appropriate for modeling our highly imbalanced dataset.

\iffalse
This paper is organised as follows.
Section 2 describes the preliminaries about the skew-elliptical and unified skew-elliptical distributions.
Section 3 details our proposed methodology.
We first introduce the new skew-elliptical link model and prove that the regression coefficients of this model have a unified skew-elliptical posterior, and then we focus on two important special cases, i.e., the skew-normal and skew-t link models.
Section 4 concerns a simulation study and an application to COVID-19 pandemic data.
Section 5 concludes with a discussion and perspectives on future research.
\fi

%%%%%%%%%%%%%%%%%%%%%%%%%%%%%%%%%%%%%%%%%%%%%%%%%%%%%%%%%%%%%%%%%%%%%
\section{Preliminaries: Skew-Elliptical and Unified Skew-Elliptical Distributions}
\subsection{The Skew-Elliptical Distribution}
The skew-elliptical distribution, originally proposed by \cite{Azzalini1999}, was formulated by multiplying an elliptical density with a skewing function.
\cite{Branco2001} proposed a new formulation of the skew-elliptical distribution by means of a conditioning mechanism.
The close relationship between these two formulations is established in \cite{Azzalini2003}.
Thanks to the construction in terms of a conditioning mechanism, the formulation in \cite{Branco2001} has led to many attractive properties of this class of distribution, such as existence of stochastic representation and closeness under marginalization and affine transformation.
\cite{Fang2003} considered a slightly wider class of distributions than \cite{Branco2001} by adding an extra truncation parameter, which was later called the extended skew-elliptical distribution in \cite{AreGenton2010}, and showed that this new distribution is closed under marginalization, affine transformation and also conditioning.
%We refer to \cite{Adcock2020} for an overview of the skew-elliptical distribution.

Here we adopt a slightly different parametrization than \cite{Fang2003} with the truncation parameter taken as $0$ and consider only skew-elliptical random vectors which possess densities.
Let $g^{(d+1)}$ be a density generator for a $(d+1)$-dimensional elliptical random vector that satisfies
\[
\int_0^\infty r^{(d+1)/2-1} g^{(d+1)}(r) {\rm d}r = \Gamma((d+1)/2)\pi^{-(d+1)/2},
\]
then a $d$-dimensional random vector $\bm{X}$ has a skew-elliptical distribution with location parameter vector $\bm{\xi}\in\mathR^d$, positive-definite scale matrix $\Sigma\in\mathR^{d\times d}$, skewness parameter vector $\bm{\alpha}\in\mathR^d$, and density generator $g^{(d+1)}$, if its density function is 
\begin{equation}\label{SEpdf}
    f_{\bm{X}}(\bm{x}) =  \frac{2}{\sqrt{|\Sigma|}}g^{(d)}\big((\bm{x}-\bm{\xi})^\top\Sigma^{-1}(\bm{x}-\bm{\xi})\big) G\big(\bm{\alpha}^\top\sigma^{-1}(\bm{x}-\bm{\xi});g_{q(\bm{x})}\big), \quad \bm{x}\in\mathR^d,
\end{equation}
where $\sigma={\rm diag}(\Sigma)^{1/2}\in\mathR^{d\times d}$, $q(\bm{x})=(\bm{x}-\bm{\xi})^\top\Sigma^{-1}(\bm{x}-\bm{\xi})$, $g^{(d)}$ is the $d$-dimensional marginal density generator induced by $g^{(d+1)}$, and $G(\cdot;g_{q(\bm{x})})$ is the cumulative distribution function of the univariate elliptical distribution with mean $0$, scale $1$, and conditional density generator $g_{q(\bm{x})}(s)=g^{(d+1)}(s+q(\bm{x}))/g^{(d)}(q(\bm{x}))$.
We write $\bm{X}\sim \mathcal{SE}_d(\bm{\xi},\Sigma,\bm{\alpha},g^{(d+1)})$.
When $\bm{\alpha}=\bm{0}$, the skew-elliptical distribution reduces to an elliptical distribution.

The skew-elliptical distribution has two stochastic representations, i.e., a convolution-type representation and a conditioning-type representation; see Equations (10) and (19) in \cite{Fang2003}. 
The former is useful for random sampling, and the latter allows us to express its cumulative distribution function in the following simple form
\begin{equation}
    F(\bm{x}) = 2G_{d+1}(\bm{x}_*-\bm{\xi}_*; \Sigma_*,g^{(d+1)}), \label{SEcdf}
\end{equation}
with $\bm{x}_* = (0, \bm{x}^\top)^\top$, $\bm{\xi}_* = (0, \bm{\xi}^\top)^\top$ and 
\[
\Sigma_* = \begin{pmatrix}
1 & -\bm{\delta}^\top \sigma \\
-\sigma\bm{\delta}  & \Sigma
\end{pmatrix},
\]
where $\sigma={\rm diag}(\Sigma)^{1/2}\in\mathR^{d\times d}$, $\bm{\delta}=(1+\bm{\alpha}^\top\Bar{\Sigma}\bm{\alpha})^{-1/2}\Bar{\Sigma}\bm{\alpha}$ with $\Bar{\Sigma}$ being the correlation matrix corresponding to $\Sigma$, i.e., $\Sigma=\sigma\Bar{\Sigma}\sigma$, and  $G_{d+1}(\bm{x}_*-\bm{\xi}_*; \Sigma_*,g^{(d+1)})$ denotes the cumulative distribution function of the $(d+1)$-variate elliptical distribution with location vector $\bm{\xi}_*\in \mathR^{d+1}$, positive-definite covariance matrix $\Sigma_*\in\mathR^{(d+1)\times (d+1)}$, and density generator $g^{(d+1)}$.
The positive definiteness of $\Sigma_*$ implies that the admissible parameters of $(\Sigma,\alpha)$ are such that the matrix $\bar\Sigma - \bm{\delta}\bm{\delta}^\top$ is positive definite.

A prominent subclass of the skew-elliptical distribution is the skew-normal distribution \citep{Azzalini1985,Azzalini1996}.
Specifically, when $g^{(d+1)}$ is the $(d+1)$-variate normal density generator, the density function of $\bm{X}$ is
\begin{equation*}\label{SNpdf}
f(\bm{x})= 2\phi_d(\bm{x}-\bm{\xi};\Sigma)\Phi\big(\bm{\alpha}^\top \sigma^{-1}(\bm{x}-\bm{\xi}) \big), \quad \bm{x}\in\mathR^d,
\end{equation*}
where $\phi_d(\bm{x}-\bm{\xi};\Sigma)$ denotes the probability density function of the $d$-variate Gaussian distribution with mean vector $\bm{\xi}$ and covariance matrix $\Sigma$, and $\Phi(\cdot)$ is the cumulative distribution function of the standard normal distribution.
We denote this distribution as $\bm{X} \sim \mathcal{SN}_d(\bm{\xi},\Sigma,\bm{\alpha})$.
When $\bm{\alpha} = \bm{0}$, it reduces to the $d$-dimensional normal distribution, $\mathcal{N}_d(\bm{\xi},\Sigma)$, and when $d=1$ it coincides with the univariate skew-normal distribution \citep{Azzalini1985}.

When $g^{(d+1)}$ is the $d$-variate Student's $t$ density generator with $\nu$ degrees of freedom, we get another important subclass of the skew-elliptical distribution, i.e., the skew-$t$ distribution \citep{Branco2001,Azzalini2003,Gupta2003}.
Its density has the following form
\begin{equation*}\label{STpdf}
    f_{\bm{X}}(\bm{x}) =  2t_d(\bm{x}-\bm{\xi};\Sigma,\nu) T\Big\{\bm{\alpha}^\top \sigma^{-1} (\bm{x}-\bm{\xi}) \Big(\frac{\nu+p}{q(\bm{x})+\nu}\Big)^{1/2}; \nu+d \Big\},\quad \bm{x}\in \mathbb{R}^d,
\end{equation*}
where $t_d(\bm{x}-\bm{\xi};\Sigma,\nu)$ denotes the probability density function of the $d$-variate t distribution with location vector $\bm{\xi}$, scale matrix $\Sigma$, and degrees of freedom $\nu$, $T(\cdot; \nu+d)$ denotes the univariate $t$ distribution function with degrees of freedom $\nu+d$.
We write $\bm{X}\sim \mathcal{ST}_d(\bm{\xi},\Sigma,\bm{\alpha},\nu)$.
When $\bm{\alpha} = \bm{0}$, it reduces to the $d$-dimensional Student's $t$ distribution, and when $\nu\rightarrow\infty$, it tends to the $d$-dimensional skew-normal distribution.

\subsection{The Unified Skew-Elliptical Distribution}
An extension of the skew-elliptical distribution is the unified skew elliptical distribution \citep{AreGenton2010}, which aims to gain more flexibility by unifying various skew-elliptical families under the same model.
Specifically, a $d$-dimensional random vector $\bm{X}$ has a unified skew elliptical distribution, denoted by $\bm{X}\sim \mathcal{SUE}_{d,m}(\bm{\xi},\Sigma,\Lambda,\bm{\tau},\Gamma,g^{(d+m)})$, if its density function is
\begin{equation*}\label{SUEpdf}
 f(\bm{x}) = \frac{g^{(d)}\big((\bm{x}-\bm{\xi})^\top\Sigma^{-1}(\bm{x}-\bm{\xi})\big)}{\sqrt{|\Sigma|}} \frac{G_m \big(\bm{\tau}+\Lambda\sigma^{-1}(\bm{x}-\bm{\xi});\Gamma, g^{(m)}_{q(\bm{x})} \big)}{G_m(\bm{\tau};\Gamma +\Lambda\Bar{\Sigma}\Lambda^\top, g^{(m)})}, \quad \bm{x}\in \mathbb{R}^d,
\end{equation*}
where $q(\bm{x})=(\bm{x}-\bm{\xi})^\top\Sigma^{-1}(\bm{x}-\bm{\xi})$, $g^{(d+m)}$ is a $(d+m)$-variate elliptical density generator, $g^{(d)}$ and $g^{(m)}$ are its $d$-variate and $m$-variate marginal density generators, respectively, $g^{(m)}_{q(\bm{x})}(s)=g^{(d+m)}\{s+q(\bm{x})\}/g^{(d)}\{q(\bm{x})\}$,
$\bm{\xi}\in\mathR^d$ is a location parameter vector, $\bm{\tau}\in\mathR^d$ introduces additional flexibility to capture skewness, $\Gamma \in\mathR^{m\times m}$ is a correlation matrix, and $\Lambda\in\mathR^{m\times d}$ encompasses the main effect on the skewness.
When $m=1$, it reduces to the extended skew-elliptical distribution \citep{Fang2003}, and if we further have $\bm{\tau}=\bm{0}$, it reduces to the skew-elliptical distribution (\ref{SEpdf}).

Similar to the skew-elliptical distribution, the unified skew elliptical distribution also has two special subclasses, i.e., the unified skew-normal distribution \citep{AreAzza2006} and the unified skew-$t$ distribution.
When $g^{(d+m)}$ is the $(d+m)$-variate normal density generator, we get the unified skew-normal distribution with density
\begin{equation}\label{SUNpdf}
 f(\bm{x}) = \phi_d(\bm{x}-\bm{\xi};\Sigma) \frac{\Phi_m \big(\bm{\tau}+\Lambda\sigma^{-1}(\bm{x}-\bm{\xi});\Gamma\big)}{\Phi_m(\bm{\tau};\Gamma +\Lambda\Bar{\Sigma}\Lambda^\top)}, \quad \bm{x}\in \mathbb{R}^d,
\end{equation}
where $\Phi_m(\cdot;\Gamma)$ denotes the centered $m$-dimensional normal distribution function with covariance matrix $\Gamma$.
We write $\bm{X}\sim \mathcal{SUN}_{d,m}(\bm{\xi},\Sigma,\Lambda,\bm{\tau},\Gamma)$.
The definition (\ref{SUNpdf}) is equivalent to the one in \cite{AreAzza2006} with a different parametrization.
When $g^{(d+m)}$ is the $(d+m)$-variate Student's $t$ density generator with $\nu$ degrees of freedom, we get the unified skew-$t$ distribution with density
\begin{equation*}\label{SUTpdf}
 f(\bm{x}) = t_d(\bm{x}-\bm{\xi};\Sigma,\nu) \frac{T_m \Big(\big(\bm{\tau}+\Lambda\sigma^{-1}(\bm{x}-\bm{\xi})\big)\big(\frac{\nu+p}{q(\bm{x})+\nu}\big)^{1/2};\Gamma,\nu+d \Big)}{T_m(\bm{\tau};\Gamma +\Lambda\Bar{\Sigma}\Lambda^\top,\nu)}, \quad \bm{x}\in \mathbb{R}^d,
\end{equation*}
where $T_m(\cdot;\Gamma,\nu+d)$ denotes the centered $m$-dimensional Student's t distribution function with dispersion matrix $\Gamma$ and degrees of freedom $\nu+d$.
We write $\bm{X}\sim \mathcal{SUT}_{d,m}(\bm{\xi},\Sigma,\Lambda,\nu,\bm{\tau},\Gamma)$.

%Note that the SUE family can be constructed by a selection mechanism \citep{AreBranco2006} and this conditioning-type stochastic representation allows us to express its cumulative distribution function in a simple way.
%Another nice property of the SUE family is that it is closed under marginalization, conditioning and affine transformations.
%We refer to \cite{AreGenton2010} for details of these properties.

%%%%%%%%%%%%%%%%%%%%%%%%%%%%%%%%%%%%%%%%%%%%%%%%%%%%%%%%%%%%%%%%%%%%%
\section{Posterior Inference for the Skew-Elliptical Link Model}
\subsection{The Skew-Elliptical Link Model}
As discussed in Section \ref{intro}, when modeling correlated binary data, the multivariate probit model uses a Gaussian link to capture dependence at the ``latent level".
A symmetric link, however, does not always provide the best fit to a given dataset, in particular for binary response data that are highly imbalanced.

In this section we extend the Gaussian link to the multivariate skew-elliptical link, which includes the skew-normal and skew-$t$ links as special cases.
Specifically, let $Y_{ij}$ denote a binary $0/1$ response on the $i$th observation of the $j$th variable and denote by $\bm{Y}_i=(Y_{i1},\dots,Y_{iM})^\top$ the collection of the $i$th observation on all $M$ variables for $i=1,\dots,n$.
Let $\bm{Y}_i^*=(Y_{i1}^*,\dots,Y_{iM}^*)^\top$ be a vector of latent variables capturing dependence among the components of $\bm{Y}_i$, $\bm{\beta}\in \mathR^p$ be a vector of regression coefficients, $X_i = (\bm{x}_{i1},\dots,\bm{x}_{iM})^\top \in \mathR^{M\times p}$ be the data matrix for the $i$th observation, and denote $X=(X_1^\top,\dots,X_n^\top)^\top \in\mathR^{nM\times p}$.
Then the multivariate skew-elliptical link model can be expressed as
\begin{align}\label{SEmodel}
%\begin{split}
    %Y_{ij}  &= \begin{cases}
    %1, \text{ if } Y_{ij}^* > 0, \\
    %0, \text{ otherwise }, 
    %\end{cases} \nonumber \\
    Y_{ij}  &= \left\{
    \begin{array}{cc}
    1, &\text{ if } Y_{ij}^* > 0, \\
    0, &\text{ otherwise}, 
    \end{array}
    \right. \nonumber \\
    \bm{Y}^* &=(\bm{Y}_1^{*\top},\dots,\bm{Y}_n^{*\top})^\top= X \bm{\beta} + \bm{\varepsilon}, \\
    \left(
    \begin{array}{c}
    \bm{\beta} \\
    \bm{\varepsilon} \\
    \end{array}
    \right) \Bigg| \Sigma,\bm{\alpha},g^{(p+nM+1)} &\sim \mathcal{SE}_{p+nM}\left(\left(
                         \begin{array}{c}
                             \bm{\mu} \\
                             0 \\
                         \end{array}
                         \right),\left(
                                 \begin{array}{cc}
                                    \Omega & 0 \\
                                    0 & \mathrm{I}_n\otimes\Sigma \\
                                \end{array}
                                 \right), \left(
                                          \begin{array}{c}
                                               \bm{0}  \\
                                               \bm{\alpha}
                                          \end{array}
                                          \right),g^{(p+nM+1)}
    \right) \nonumber,
%\end{split}
\end{align}
where $\bm{\mu}\in\mathR^p$ is a location parameter vector, $\Omega\in\mathR^{p\times p}$ is a positive-definite covariance matrix, $\mathrm{I}_n\in\mathR^{n\times n}$ is the identity matrix, $\otimes$ denotes the Kronecker product, $\Sigma\in\mathR^{M\times M}$ is a positive definite covariance matrix, $\bm{\alpha}\in\mathR^{nM}$ is a skewness parameter vector, and $g^{(p+nM+1)}$ is a $(p+nM+1)$-variate elliptical density generator.

The multivariate probit model \citep{Chib1998} assumes that the covariates are not shared by the $M$ variables $Y_{i1},\dots,Y_{iM}$.
In that case, $\bm{\beta}$ can be understood as $\bm{\beta}=(\bm{\beta}_1^\top,\dots,\bm{\beta}_M^\top)^\top$, where $\bm{\beta}_j\in\mathR^{p_j}$ with $\sum_{j=1}^M p_j = p$ is the regression coefficients for the $j$-th variable $Y_{1j},\dots,Y_{nj}$, and $\bm{x}_{ij}$ is understood as the vector $\bm{x}_{ij}=(\bm{x}_{ij1}^\top,\dots,\bm{x}_{ijM}^\top)^\top$ with $\bm{x}_{ijk}=\bm{0}$ for $k\neq j$, so that $\bm{x}_{ij}^\top \bm{\beta} = \bm{x}_{ijj}^\top \bm{\beta}_j$.
This notation of expanded vector $\bm{\beta}$ and $\bm{x}_{ij}$ simplifies the expression of our model (\ref{SEmodel}).

To better understand the assumption on the joint distribution of $\bm{\beta}$ and $\bm{\varepsilon}$ in the model (\ref{SEmodel}), we express it in a different way.
Using Proposition 2 in \cite{Fang2003}, an equivalent assumption is that
\begin{align}
    \bm{\beta} \mid g^{(p+nM+1)} &\sim \mathcal{SE}_{p}(\bm{\mu},\Omega,\bm{0},g^{(p+1)}) \label{betaprior},\\
    \bm{\varepsilon}\mid \bm{\beta},\Sigma,\bm{\alpha},g^{(p+nM+1)} &\sim \mathcal{SE}_{nM}(\bm{0},\mathrm{I}_n\otimes\Sigma,\bm{\alpha},g^{(nM+1)}_{q(\bm{\beta})}), \label{modelassumption}
\end{align}
where $q(\bm{\beta})=(\bm{\beta}-\bm{\mu})^\top\Omega^{-1}(\bm{\beta}-\bm{\mu})$, $g^{(nM+1)}_{q(\bm{\beta})}(s) = g^{(p+nM+1)}(s+q(\bm{\beta}))/g^{(p)}(q(\bm{\beta}))$, and $g^{(p)}$, $g^{(p+1)}$ are the $p$- and $(p+1)$-variate marginal density generators induced by the same generator $g^{(p+nM+1)}$, respectively.
Assumption (\ref{betaprior}) may be understood as the prior for $\bm{\beta}$, while (\ref{modelassumption}) is the distributional assumption for the latent data vector $\bm{Y}^*$.
From (\ref{modelassumption}) we observe that $\beta$ and $\varepsilon$ are dependent, but they are conditionally independent given $q(\beta)$.
This weak dependence between them is broken when $g^{(p+nM+1)}$ is the normal density generator.
Specifically, when $g^{(p+nM+1)}$ is the $(p+nM+1)$-variate normal density generator, (\ref{betaprior}) becomes the typical Gaussian prior, $\mathcal{N}_p(\bm{\mu},\Omega)$, and (\ref{modelassumption}) becomes 
$\bm{\varepsilon}\mid \Sigma,\bm{\alpha} \sim \mathcal{SN}_{nM}(\bm{0},\mathrm{I}_n\otimes\Sigma,\bm{\alpha})$, which is independent of $\bm{\beta}$ conditional on $\Sigma$ and $\bm{\alpha}$.
If we further have $\bm{\alpha}=\bm{0}$, then (\ref{modelassumption}) becomes $\bm{\varepsilon}\mid \Sigma \sim \mathcal{N}_{nM}(\bm{0},\mathrm{I}_n\otimes\Sigma)$ and model (\ref{SEmodel}) reduces to the well-known multivariate probit model \citep{Ashford1970,Chib1998} with a typical Gaussian prior for $\bm{\beta}$.
By assuming a joint distribution for $\bm{\beta}$ and $\bm{\varepsilon}$, we can gain two major advantages.
The first is that we are able to account not only for the dependence between $\bm{\beta}$ and $\bm{\varepsilon}$, but also for the dependence between the different observations $\bm{Y}_i,i=1,\dots,n$.
The second is that this assumption allows us to get a tractable posterior for $\bm{\beta}$; see Section \ref{regsampling} for more details.

From (\ref{modelassumption}) we know that the admissible parameters of $(\Sigma,\alpha)$ are these such that the matrix $\mathrm{I}_n\otimes\bar\Sigma - \bm{\delta}\bm{\delta}^\top$ is positive definite, where $\bar\Sigma$ is the correlation matrix corresponding to $\Sigma$ and $\bm{\delta}=\big(1+\bm{\alpha}^\top(\mathrm{I}_n\otimes\bar\Sigma)\bm{\alpha}\big)^{-1/2}(\mathrm{I}_n\otimes\bar\Sigma)\bm{\alpha}$.
From (\ref{SEmodel}), the joint probability mass function of $\bm{Y}=(\bm{Y}_1^\top,\dots,\bm{Y}_n^\top)^\top=\bm{y}$, given all the parameters and the data matrix $X$, is
\begin{equation}\label{jointprob}
%\begin{split}
p(\bm{y} \mid \bm{\beta},\Sigma,\bm{\alpha},g^{(p+nM+1)})
=\int_{A_{nM}}\cdots\int_{A_{11}} \frac{2}{|\mathrm{I}_n\otimes\Sigma|^{1/2}}g^{q(\bm{\beta}),nM}\big(\bm{t}^\top(\mathrm{I}_n\otimes\Sigma^{-1})\bm{t}\big) G(\bm{\alpha}^\top\bm{t};g^{q(\bm{\beta})}_{q(\bm{t})})  {\rm d}\bm{t},
%\end{split}
\end{equation}
where $q(\bm{t})=\bm{t}^\top(\mathrm{I}_n\otimes\Sigma^{-1})\bm{t}$, $g^{q(\bm{\beta})}_{q(\bm{t})}(s) = g^{(nM+1)}_{q(\bm{\beta})}(s+q(\bm{t}))/g^{q(\bm{\beta}),nM}(q(\bm{t}))$, $g^{q(\bm{\beta}),nM}$ is the $nM$-variate marginal density generator induced by $g^{(nM+1)}_{q(\bm{\beta})}$, and $A_{ij}, i=1,\dots,n, j=1,\dots,M$ is the interval
\[
%A_{ij} = \begin{cases}
%(-\bm{x}_{ij}^\top \bm{\beta},\infty), &\text{ if } y_{ij}=1, \\
%(-\infty,\bm{x}_{ij}^\top \bm{\beta}], &\text{ if } y_{ij}=0.
%\end{cases}
A_{ij} = \left\{
\begin{array}{cc}
(-\bm{x}_{ij}^\top \bm{\beta},\infty), &\text{ if } y_{ij}=1, \\
(-\infty,\bm{x}_{ij}^\top \bm{\beta}], &\text{ if } y_{ij}=0.
\end{array}
\right.
\]
Although the joint probability (\ref{jointprob}) involves multidimensional integration over a constrained space, we show in the following section that it can be substantially simplified.
%Similarly as in the multivariate probit model, the matrix $\Bar{\Sigma}$ has to be a correlation matrix for identifiability reasons.
%This can be seen by taking $\Tilde{\bm{\beta}} = {\rm diag}(\mathrm{I}_n\otimes\Sigma)^{1/2}\bm{\beta}$, and $\Tilde{\bm{\alpha}} = {\rm diag}(\mathrm{I}_n\otimes\Sigma^{-1})^{1/2}\bm{\alpha}$, which gives $p(\bm{y} \mid \bm{\beta},\Bar{\Sigma},\bm{\alpha},g^{(p+nM+1)}) = p(\bm{y} \mid \Tilde{\bm{\beta}},\Sigma,\Tilde{\bm{\alpha}},g^{(p+nM+1)})$.

\subsection{Unified Skew Elliptical Posterior for the Regression Coefficients}\label{regsampling}
In this section we prove that for the multivariate skew-elliptical link model (\ref{SEmodel}), the regression coefficients parameter $\bm{\beta}$ has a unified skew elliptical posterior.
To prove this result, we first simplify the joint probability mass function $p(\bm{y} \mid \bm{\beta},\Sigma,\bm{\alpha},g^{(p+nM+1)})$ of the observed data in the following lemma. 
%All proofs are deferred to the Appendix.
\begin{lemma}\label{lemma1}
The joint probability mass function $p(\bm{y} \mid \bm{\beta},\Sigma,\bm{\alpha},g^{(p+nM+1)})$ based on (\ref{SEmodel}) can be simplified to
\[
p(\bm{y} \mid \bm{\beta},\Sigma,\bm{\alpha},g^{(p+nM+1)}) = 2G_{nM+1}(D_*\bm{\beta}; \Sigma_*,g^{(nM+1)}_{q(\bm{\beta})}), 
\]
where $D = {\rm diag}(2\bm{y}-\bm{1}_{nM})\in\mathR^{nM\times nM}$ with $\bm{1}_{nM} \in \mathR^{nM}$ being the vector of $1$s, $D_*=\big(\bm{0}_p,(DX)^\top\big)^\top\in\mathR^{(nM+1)\times p}$, 
$\bm{0}_p\in\mathR^p$ is a vector of $0$s, and 
\[
\Sigma_* = \begin{pmatrix}
1 & -\bm{\delta}^\top D(\mathrm{I}_n\otimes\sigma) \\
-(\mathrm{I}_n\otimes\sigma)D\bm{\delta} & D(\mathrm{I}_n\otimes\Sigma)D
\end{pmatrix}\in\mathR^{(nM+1)\times (nM+1)}
\]
with $\bm{\delta}\in\mathR^{nM}, \bm{\delta}=\big(1+\bm{\alpha}^\top(\mathrm{I}_n\otimes\bar\Sigma)\bm{\alpha}\big)^{-1/2}(\mathrm{I}_n\otimes\bar\Sigma)\bm{\alpha}$, $\sigma={\rm diag}(\Sigma)^{1/2}\in\mathR^{d\times d}$ and $\bar\Sigma$ being the correlation matrix corresponding to $\Sigma$, i.e., $\Sigma=\sigma\bar\Sigma\sigma$.
\end{lemma}
\begin{proof}
Since a diagonal matrix ${\rm diag}(\bm{x})$ with $\bm{x}\in\{-1,1\}^{nM}$ has the property
\[
{\rm diag}(\bm{x}) \bm{x} = \bm{1}_{nM}, \text{ and } \big({\rm diag}(\bm{x})\big)^{-1} = {\rm diag}(\bm{x}),
\]
we have
\begin{align*}
    p(\bm{y} \mid \bm{\beta},\Sigma,\bm{\alpha},g^{(p+nM+1)}) &= \mathP(\bm{Y}=\bm{y} \mid \bm{\beta},\Sigma,\bm{\alpha},g^{(p+nM+1)}) \\
    &= \mathP(2\bm{Y}-\bm{1}_{nM}=2\bm{y}-\bm{1}_{nM} \mid \bm{\beta},\Sigma,\bm{\alpha},g^{(p+nM+1)}) \\
    &= \mathP\big(D(2\bm{Y}-\bm{1}_{nM})=\bm{1}_{nM} \mid \bm{\beta},\Sigma,\bm{\alpha},g^{(p+nM+1)}\big) \\
    &= \mathP\big(D\bm{Y}^* > \bm{0} \mid \bm{\beta},\Sigma,\bm{\alpha},g^{(p+nM+1)}\big) \\
    &= \mathP\big(-D\bm{\varepsilon}-D X\bm{\beta} < \bm{0} \mid \bm{\beta},\Sigma,\bm{\alpha},g^{(p+nM+1)}\big).
\end{align*}
By (\ref{modelassumption}), $\bm{\varepsilon}\mid \bm{\beta},\Sigma,\bm{\alpha},g^{(p+nM+1)} \sim \mathcal{SE}_{nM}(\bm{0},\mathrm{I}_n\otimes\Sigma,\bm{\alpha},g^{(nM+1)}_{q(\bm{\beta})})$. 
Using Proposition 1 in \cite{Fang2003}, we know that
\[
(-D\bm{\varepsilon}-D X\bm{\beta}) \mid \bm{\beta},\Sigma,\bm{\alpha},g^{(p+nM+1)} \sim \mathcal{SE}_{nM}(-D X\bm{\beta},D(\mathrm{I}_n\otimes\Sigma)D,D\bm{\alpha},g^{(nM+1)}_{q(\bm{\beta})}).
\]
Using (\ref{SEcdf}), we finally get
\[
p(\bm{y} \mid \bm{\beta},\Sigma,\bm{\alpha},g^{(p+nM+1)}) = 2G_{nM+1}(D_*\bm{\beta}; \Sigma_*,g^{(nM+1)}_{q(\bm{\beta})}).
\]
\end{proof}

Now we are ready to present our main result that the posterior distribution of $\bm{\beta}$ coincides with a unified skew elliptical distribution.

\begin{theorem}\label{theorem1}
Let $\bm{y}=(\bm{y}_1^\top,\dots,\bm{y}_n^\top)^\top$ be observations from the multivariate skew-elliptical link model (\ref{SEmodel}) and $X=(X_1^\top,\dots,X_n^\top)^\top$ be the corresponding data matrix.
Then
\[
(\bm{\beta} \mid \bm{y},\Sigma,\bm{\alpha},g^{(p+nM+1)}) \sim \mathcal{SUE}_{p,nM+1}(\bm{\mu}_{post},\Omega_{post},\Lambda_{post},\bm{\tau}_{post},\Gamma_{post},g^{(p+nM+1)}),
\]
with posterior parameters
\begin{align*}
    \bm{\mu}_{post} = \bm{\mu}, \; \Omega_{post}=\Omega, \; \Lambda_{post}=\sigma_*^{-1}D_*\omega, \; \bm{\tau}_{post}=\sigma_*^{-1}D_*\bm{\mu}, \; \Gamma_{post}=\bar\Sigma_*,
\end{align*}
where $D_*\in\mathR^{(nM+1)\times p}$ and $\Sigma_*\in\mathR^{(nM+1)\times (nM+1)}$ are the matrices defined in Lemma \ref{lemma1}, $\sigma_*={\rm diag}(\Sigma_*)^{1/2}\in\mathR^{(nM+1)\times (nM+1)}$, $\bar\Sigma_*$ is the correlation matrix corresponding to $\Sigma_*$, i.e., $\Sigma_*=\sigma_*\bar\Sigma_*\sigma_*$,
and $\omega={\rm diag}(\Omega)^{1/2}\in\mathR^{p\times p}$.
\end{theorem}
\begin{proof}
The posterior density of the coefficients $\bm{\beta}$ is 
\begin{equation*}
    p(\bm{\beta} \mid \bm{y},\Sigma,\bm{\alpha},g^{(p+nM+1)}) \propto p(\bm{y} \mid \bm{\beta},\Sigma,\bm{\alpha},g^{(p+nM+1)}) \cdot p(\bm{\beta}\mid g^{(p+nM+1)}).
\end{equation*}
Using Lemma \ref{lemma1} and the assumption (\ref{betaprior}), we have
\begin{align*}
    &\quad p(\bm{\beta} \mid \bm{y},\Sigma,\bm{\alpha},g^{(p+nM+1)}) \\
    &\propto G_{nM+1}(D_*\bm{\beta}; \Sigma_{*},g^{(nM+1)}_{q(\bm{\beta})}) \cdot g^{(p)}\big((\bm{\beta}-\bm{\mu})^\top\Omega^{-1}(\bm{\beta}-\bm{\mu})\big) \\
    &= G_{nM+1}(\sigma_*^{-1}D_*\bm{\beta}; \bar\Sigma_{*},g^{(nM+1)}_{q(\bm{\beta})}) \cdot g^{(p)}\big((\bm{\beta}-\bm{\mu})^\top\Omega^{-1}(\bm{\beta}-\bm{\mu})\big) \\
    &= G_{nM+1}\big(\sigma_*^{-1}D_*\bm{\mu}+\sigma_*^{-1}D_*(\bm{\beta}-\bm{\mu}); \bar\Sigma_{*},g^{(nM+1)}_{q(\bm{\beta})}\big) \cdot g^{(p)}\big((\bm{\beta}-\bm{\mu})^\top\Omega^{-1}(\bm{\beta}-\bm{\mu})\big) \\
    &= G_{nM+1}\big(\sigma_*^{-1}D_*\bm{\mu}+\sigma_*^{-1}D_* \omega \omega^{-1}(\bm{\beta}-\bm{\mu}); \bar\Sigma_{*},g^{(nM+1)}_{q(\bm{\beta})}\big) \cdot g^{(p)}\big((\bm{\beta}-\bm{\mu})^\top\Omega^{-1}(\bm{\beta}-\bm{\mu})\big) \\
    &= G_{nM+1}(\bm{\tau}_{\rm post}+\Lambda_{\rm post} \omega^{-1}(\bm{\beta}-\bm{\mu}); \bar\Sigma_{*},g^{(nM+1)}_{q(\bm{\beta})}) \cdot g^{(p)}\big((\bm{\beta}-\bm{\mu})^\top\Omega^{-1}(\bm{\beta}-\bm{\mu})\big).
\end{align*}
Hence, $(\bm{\beta} \mid \bm{y},\Sigma,\bm{\alpha},g^{(p+nM+1)}) \sim \mathcal{SUE}_{p,nM+1}(\bm{\mu}_{post},\Omega_{post},\Lambda_{post},\bm{\tau}_{post},\Gamma_{post},g^{(p+nM+1)})$.
\end{proof}

%Comparing the results in Theorem \ref{theorem1} and Equation (\ref{SUEpdf}), we can derive the density function of the posterior of the regression coefficients, after minor mathematical simplifications, as
%\[
%p(\bm{\beta} \mid \bm{y},\Sigma,\bm{\alpha},g^{(p+nM+1)}) =  \frac{g^{(p)}\big((\bm{\beta}-\bm{\mu})^\top\Omega^{-1}(\bm{\beta}-\bm{\mu})\big)G_{nM+1}(D_*\bm{\beta};\Sigma_{*},g^{(nM+1)}_{q(\bm{\beta})})}{|\Omega|^{1/2}G_{nM+1}(D_*\bm{\mu}; D_*\Omega D_*+\Sigma_{*},g^{(nM+1)})},
%\]
%where $g^{(nM+1)}$ is the $(nM+1)$-variate density generator induced by $g^{(p+nM+1)}$, and \\ $|\Omega|^{1/2}G_{nM+1}(D_*\bm{\mu}; D_*\Omega D_*+\Sigma_{*},g^{(nM+1)})$ is the normalizing constant.
In Bayesian regression we are mostly interested in the posterior marginals, their moments and more complex functionals such as measures of dependence and credible intervals. 
Thanks to the fundamental property of the unified skew elliptical distribution that it is closed under marginalization, conditioning and affine transformations, this type of inference is simplified.
We refer to \cite{AreGenton2010} for details on how to obtain the parameters of the marginal distribution, conditional distribution and the distribution after affine transformations.
As for the calculation of the posterior moments and credible intervals, numerical integration of the marginal posterior densities can be used.
When interest is in the posterior moments, another approach is to use the moment generating function.
We refer to Section 5 of \cite{AreGenton2010} for derivations of the moment generating function and moments of the unified skew elliptical distribution.

\subsection{Special Case 1: the Skew-Normal Link Model}
The skew-normal link model is obtained when $g^{(p+nM+1)}$ in model (\ref{SEmodel}) is the $(p+nM+1)$-variate normal density generator.
In this case, the joint distributional assumption of $\bm{\beta}$ and $\bm{\varepsilon}$ becomes
\[
\left(
    \begin{array}{c}
    \bm{\beta} \\
    \bm{\varepsilon} \\
    \end{array}
    \right) \Bigg| \Sigma,\bm{\alpha} \sim \mathcal{SN}_{p+nM}\left(\left(
                         \begin{array}{c}
                             \bm{\mu} \\
                             0 \\
                         \end{array}
                         \right),\left(
                                 \begin{array}{cc}
                                    \Omega & 0 \\
                                    0 & I_n\otimes\Sigma \\
                                \end{array}
                                 \right), \left(
                                          \begin{array}{c}
                                               \bm{0}  \\
                                               \bm{\alpha}
                                          \end{array}
                                          \right)
    \right),
\]
which is equivalent to assuming
\[
\bm{\beta} \sim \mathcal{N}_p(\bm{\mu},\Omega), \quad \bm{\varepsilon}\mid \Sigma,\bm{\alpha} \sim \mathcal{SN}_{nM}(\bm{0},\mathrm{I}_n\otimes\Sigma,\bm{\alpha}),
\]
with the random vectors $\bm{\beta}$ and $\bm{\varepsilon}$ independent of each other given $\Sigma$ and $\bm{\alpha}$.
This implies that the prior for $\bm{\beta}$ coincides with the typical weakly informative Gaussian prior, and we use a multivariate SN distribution to model the dependence of the data at the latent level.
When the skewness parameter $\bm{\alpha}=\bm{0}$, the skew-normal link model reduces to the well-known multivariate probit model \citep{Ashford1970,Chib1998}.

Before analyzing the posterior of the regression coefficients, we first give the explicit expression of the joint probability mass function $p(\bm{y} \mid \bm{\beta},\Sigma,\bm{\alpha})$.
Taking $g^{(nM+1)}$ in Lemma \ref{lemma1} as the normal density generator, we directly get
\[
p(\bm{y} \mid \bm{\beta},\Sigma,\bm{\alpha}) = 2\Phi_{nM+1}(D_*\bm{\beta}; \Sigma_*), 
\]
where $D_*$ and $\Sigma_*$ are defined in Lemma \ref{lemma1}.
Similarly to the multivariate probit model, if the covariates are not shared by all the responses, the matrix $\Sigma$ has to be a correlation matrix for identifiability reasons.
This can be seen by considering $\Omega=(\omega_{jm})=C\Sigma C^\top$ with $C={\rm diag}(\Omega)^{1/2}$, $\Tilde{\bm{\beta}} = (\Tilde{\bm{\beta}}_1^\top,\dots,\Tilde{\bm{\beta}}_M^\top)^\top$ with $\Tilde{\bm{\beta}}_j=\omega_{jj}\bm{\beta}_j$, and  $\Tilde{\bm{\alpha}} = {\rm diag}(\mathrm{I}_n\otimes C^{-1})\bm{\alpha}$, which gives $p(\bm{y} \mid \bm{\beta},\Sigma,\bm{\alpha}) = p(\bm{y} \mid \Tilde{\bm{\beta}},\Omega,\Tilde{\bm{\alpha}})$.
If the covariates are shared by all the responses $\bm{Y}$, then $\Sigma$ does not need to be a correlation matrix.
However, if we assume that all the diagonal entries in $\Sigma$ are equal, then $\Sigma$ has to be a correlation matrix because $p(\bm{y} \mid \bm{\beta},\Sigma,\bm{\alpha}) = p(\bm{y} \mid b\bm{\beta},b^2\Sigma,\bm{\alpha})$ for any positive number $b$.

When $\bm{\alpha}=\bm{0}$, we get $p(\bm{y} \mid \bm{\beta},\Sigma)=\Phi_{nM}(D X\bm{\beta}; \mathrm{I}_n\otimes\Sigma)$.
This result simplifies the calculation of the joint probability of the multivariate skew-normal link model and also the probit model by expressing it in terms of the multivariate normal distribution function.
Therefore, existing fast algorithms for the calculation of the multivariate normal probabilities can be utilized especially in high dimensions; see \cite{Genton2018} and \cite{Cao2020}.
Now we present the result that for the skew-normal link model the posterior of $\bm{\beta}$ coincides with a unified skew normal distribution, which directly follows from Theorem \ref{theorem1} by taking $g^{(nM+1)}$ as the $(nM+1)$-variate normal density generator.

\begin{corollary}\label{SNposterior}
Let $\bm{y}=(\bm{y}_1^\top,\dots,\bm{y}_n^\top)^\top$ be observations from the multivariate skew-normal link model and $X=(X_1^\top,\dots,X_n^\top)^\top$ be the corresponding data matrix.
Then
\[
(\bm{\beta} \mid \bm{y},\Sigma,\bm{\alpha}) \sim \mathcal{SUN}_{p,nM+1}(\bm{\mu}_{post},\Omega_{post},\Lambda_{post},\bm{\tau}_{post},\Gamma_{post}),
\]
where $\bm{\mu}_{post}, \Omega_{post}, \Lambda_{post}, \bm{\tau}_{post}, \Gamma_{post}$ are defined in Theorem \ref{theorem1}.
\end{corollary}

The unified skew normal distribution, a subclass of the unified skew elliptical family, is closed under marginalization, conditioning and affine transformations \citep{AreAzza2006,AreGenton2010}.
This property is useful for certain posterior inferences, such as the posterior marginals or their moments.
%Here we give the explicit moment generating function of the posterior of the regression coefficients, which can also be used to derive its moments.
%Using Equation (16) in \cite{AreAzza2006}, after a change of parametrization, the moment generating function of $(\bm{\beta} \mid \bm{y},\Sigma,\bm{\alpha})$, which has the posterior distribution described in Corollary \ref{SNposterior}, is
%\begin{align*}
%    M(\bm{t}) 
%    = \exp\big(\bm{\mu}^\top \bm{t} + \frac{1}{2}\bm{t}^\top\Omega \bm{t}\big) \frac{\Phi_{nM+1}\big(D_*\bm{\mu}+D_*\Omega\bm{t}; D_*\Omega D_*^\top + \Sigma_*\big)}{\Phi_{nM+1}\big(D_*\bm{\mu}; D_*\Omega D_*^\top + \Sigma_*\big)}.
%\end{align*}
When interest is in sampling from the posterior distribution, the convolution-type stochastic representation of the unified skew normal random vector is very useful. Specifically, using Equation (8) in \cite{AreGenton2010}, $(\bm{\beta} \mid \bm{y},\Sigma,\bm{\alpha})$ has the following stochastic representation 
\begin{equation*}\label{SUNstochstic}
(\bm{\beta} \mid \bm{y},\Sigma,\bm{\alpha}) \equdist \bm{\mu} + \bm{V}_0 + \Omega D_*^\top (D_*\Omega D_*^\top + \Sigma_*)^{-1} s \bm{V}_1,
\end{equation*}
where $\equdist$ means equality in distribution, $s={\rm diag}(D_*\Omega D_*^\top+\Sigma_*)^{1/2}\in\mathR^{(nM+1)\times (nM+1)}$, $\bm{V}_0 \sim \mathcal{N}_p\big(\bm{0},\Omega-\Omega D_*^\top (D_*\Omega D_*^\top + \Sigma_*)^{-1} D_* \Omega\big)$ is independent of $\bm{V}_1$, which follows a $(nM+1)$-variate truncated normal distribution with location parameter $\bm{0}$, covariance matrix $s^{-1}(D_*\Omega D_*^\top + \Sigma_*)s^{-1}$ and truncated below the level $-s^{-1}D_*\bm{\mu}$.
This stochastic representation facilitates exact simulation from the posterior distribution; see Algorithm 1 of \cite{Durante2019}.

\subsection{Special Case 2: the Skew-$t$ Link Model}\label{STcase}
When $g^{(p+nM+1)}$ in model (\ref{SEmodel}) is the $(p+nM+1)$-variate Student's $t$ density generator with $\nu$ degrees of freedom, we get the skew-$t$ link model.
Specifically, the joint distributional assumption of $\bm{\beta}$ and $\bm{\varepsilon}$ is
\[
\left(
    \begin{array}{c}
    \bm{\beta} \\
    \bm{\varepsilon} \\
    \end{array}
    \right) \Bigg| \Sigma,\bm{\alpha},\nu \sim \mathcal{ST}_{p+nM}\left(\left(
                         \begin{array}{c}
                             \bm{\mu} \\
                             0 \\
                         \end{array}
                         \right),\left(
                                 \begin{array}{cc}
                                    \Omega & 0 \\
                                    0 & I_n\otimes\Sigma \\
                                \end{array}
                                 \right), \left(
                                          \begin{array}{c}
                                               \bm{0}  \\
                                               \bm{\alpha}
                                          \end{array}
                                          \right), \nu
    \right),
\]
which is equivalent to assuming
\begin{align*}
    \bm{\beta}\mid\nu &\sim \mathcal{T}_p(\bm{\mu},\Omega,\nu),\\
    \bm{\varepsilon}\mid \bm{\beta},\Sigma,\bm{\alpha},\nu &\sim \mathcal{ST}_{nM}\Bigg(\bm{0},\frac{\nu+(\bm{\beta}-\bm{\mu})^\top\Omega^{-1}(\bm{\beta}-\bm{\mu})}{\nu+p}\big(\mathrm{I}_n\otimes\Sigma\big),\bm{\alpha}, \nu+p \Bigg),
\end{align*}
where $\mathcal{T}_p(\bm{\mu},\Omega,\nu)$ denotes the Student's $t$ distribution with location parameter vector $\bm{\mu}$, dispersion matrix $\Omega$ and degrees of freedom $\nu$.
The nonnegative parameter $\nu$ can be considered as a hyper-parameter which controls the dependence between $\bm{\beta}$ and $\bm{\varepsilon}$.
As $\nu$ increases the dependence decreases, and when $\nu\rightarrow\infty$, the skew-$t$ link model tends to the skew-normal link model and the dependence between them vanishes.

By taking $g^{(nM+1)}$ in Lemma \ref{lemma1} as the Student's $t$ density generator with $\nu$ degrees of freedom, we get the following explicit expression of the joint probability
\begin{equation}\label{STprob}
p(\bm{y} \mid \bm{\beta},\Sigma,\bm{\alpha},\nu) = 2T_{nM+1}\Bigg(\bigg(\frac{\nu+p}{\nu+(\bm{\beta}-\bm{\mu})^\top\Omega^{-1}(\bm{\beta}-\bm{\mu})}\bigg)^{1/2}D_*\bm{\beta}; \Sigma_*,\nu+p \Bigg).
\end{equation}
In practice, we typically assume a weakly informative prior for $\bm{\beta}$, which means $\nu$ is often large and $\Omega$ is often taken as a diagonal matrix with large diagonal entries.
This implies that $(\bm{\beta}-\bm{\mu})^\top\Omega^{-1}(\bm{\beta}-\bm{\mu})$ is often very small compared to $\nu$ and $\nu \approx \nu + (\bm{\beta}-\bm{\mu})^\top\Omega^{-1}(\bm{\beta}-\bm{\mu})$.
Hence, if we assume that the diagonal entries of $\Sigma$ are all equal, then $\Sigma$ needs to be a correlation matrix because $p(\bm{y} \mid \bm{\beta},\Sigma,\bm{\alpha},\nu) \approx p(\bm{y} \mid b\bm{\beta},b^2\Sigma,\bm{\alpha},\nu)$ for any positive number $b$.
We now state the result that for the skew-$t$ link model the posterior of $\bm{\beta}$ coincides with a unified skew $t$ distribution, which directly follows from Theorem \ref{theorem1} by taking $g^{(nM+1)}$ as the $(nM+1)$-variate Student's $t$ density generator with $\nu$ degrees of freedom.

\begin{corollary}\label{STposterior}
Let $\bm{y}=(\bm{y}_1^\top,\dots,\bm{y}_n^\top)^\top$ be observations from the multivariate skew-t link model and $X=(X_1^\top,\dots,X_n^\top)^\top$ be the corresponding data matrix.
Then
\[
(\bm{\beta} \mid \bm{y},\Sigma,\bm{\alpha},\nu) \sim \mathcal{SUT}_{p,nM+1}(\bm{\mu}_{post},\Omega_{post},\Lambda_{post},\nu,\bm{\tau}_{post},\Gamma_{post}),
\]
where $\bm{\mu}_{post}, \Omega_{post}, \Lambda_{post}, \bm{\tau}_{post}, \Gamma_{post}$ are defined in Theorem \ref{theorem1}.
\end{corollary}

Similarly to the unified skew normal distribution, the unified skew $t$ distribution is also closed under marginalization, conditioning and affine transformations \citep{AreGenton2010}, which simplifies the inference of the posterior marginals, their moments and functionals such as measures of dependence and credible intervals.
Thanks to the stochastic representation of the unified skew $t$ distribution, exact sampling from the distribution of $(\bm{\beta} \mid \bm{y},\Sigma,\bm{\alpha},\nu)$ is also feasible.
Specifically, using Equation (9) in \cite{AreGenton2010}, $(\bm{\beta} \mid \bm{y},\Sigma,\bm{\alpha},\nu)$ has the stochastic representation 
\begin{equation}\label{SUTstochstic}
(\bm{\beta} \mid \bm{y},\Sigma,\bm{\alpha},\nu)
\equdist \bm{\mu} + \bigg(\frac{\nu+\bm{U}_1^\top s(D_*\Omega D_*^\top + \Sigma_*)^{-1}s\bm{U}_1}{\nu+nM+1}\bigg)^{1/2}\bm{U}_0 + \Omega D_*^\top (D_*\Omega D_*^\top + \Sigma_*)^{-1} s \bm{U}_1,
\end{equation}
where $s={\rm diag}(D_*\Omega D_*^\top+\Sigma_*)^{1/2}\in\mathR^{(nM+1)\times (nM+1)}$, $\bm{U}_0 \sim \mathcal{T}_p\big(\bm{0},\Omega-\Omega D_*^\top (D_*\Omega D_*^\top + \Sigma_*)^{-1} D_* \Omega,\nu+nM+1\big)$ is independent of $\bm{U}_1$, which follows a $(nM+1)$-variate truncated $t$ distribution with location parameter vector $\bm{0}$, dispersion matrix $s^{-1}(D_*\Omega D_*^\top + \Sigma_*)s^{-1}$, degrees of freedom $\nu$, and truncated below the level $-s^{-1}D_*\bm{\mu}$.

%%%%%%%%%%%%%%%%%%%%%%%%%%%%%%%%%%%%%%%%%%%%%%%%%%%%%%%%%%%%%%%%%%%%%
\section{Simulation and Empirical Studies}
\subsection{Prior and Posterior for $\mathbb{\alpha}$ and $\Sigma$}\label{covsampling}
As the skew-normal link model is a limiting case of the skew-$t$ link model when the degrees of freedom $\nu$ tends to $\infty$, in this section we focus on the skew-$t$ link model and perform a simulation study and a real-data application.
To make the model parsimonious, in both the simulation study and empirical study we assume that the skewness parameters are the same across different observations, i.e., $\bm{\alpha}=(\alpha_1,\dots,\alpha_M,\dots,\alpha_1,\dots,\alpha_M)^\top\in\mathR^{nM}$, and $\Sigma$ is a correlation matrix, i.e., $\Sigma=\bar\Sigma$.
The assumption of a correlation matrix for $\Sigma$ is not very restrictive because it is approximately equivalent to assuming that all the diagonal entries in $\Sigma$ are equal, as we discussed in Section \ref{STcase}.
Now we specify the prior and posterior for the skewness parameter $\bm{\alpha}_s= (\alpha_1,\dots,\alpha_M)^\top$ and the correlation matrix $\bar\Sigma$.

Bayesian modeling of unstructured covariance or correlation matrices is a fundamental and difficult task because of the constraint of positive definiteness and the quadratic increase of the number of parameters with respect to the number of correlated variables.
More importantly, it is difficult to specify a prior for them\citep{Gelman2014}.
Typical priors for correlation matrices include the marginally uniform prior, the jointly uniform prior \citep{Barnard2000} and the so-called LKJ prior \citep{LKJ2009}.

The marginally uniform prior means that each non-diagonal element in the correlation matrix has a uniform marginal distribution over $[-1,1]$, whereas the jointly uniform prior means that the correlation matrix has a joint uniform distribution over the compact space of valid correlation matrices.
%The jointly uniform prior does not imply that each non-diagonal entry $\sigma_{ij},i\neq j$, of the correlation matrix has a uniform marginal distribution, but rather it favors values of $\sigma_{ij}$ close to zero over values close to $\pm 1$, in particular in high dimensions.
%For more discussion about the shape of the jointly uniform prior, we refer to \cite{Barnard2000} and references therein.
The LKJ prior is recommended in the \texttt{R} library \texttt{rstan} \citep{R2020} and has the form $\pi(\Bar{\Sigma}) \propto |\Bar{\Sigma}|^{\eta-1}$, where $|\Bar{\Sigma}|$ is the determinant of $\Bar{\Sigma}$ and $\eta > 0$ is the shape parameter of the LKJ distribution. 
The jointly uniform prior is a special case of the LKJ prior when $\eta = 1$.

In this work we adopt the jointly uniform prior for $\Bar{\Sigma}$ by setting $\eta = 1$ in the LKJ prior and specify an independent weakly informative Gaussian prior for $\bm{\alpha}_s$. 
Then, using Equation (\ref{STprob}), the joint posterior of $(\Bar{\Sigma},\bm{\alpha}_s)$ given the data and the regression coefficients is
\begin{align}\label{covposterior}
    p(\Bar{\Sigma},\bm{\alpha}_s \mid \bm{y},\bm{\beta},\nu) 
    %&\propto p(\bm{y} \mid \bm{\beta}, \Bar{\Sigma},\bm{\alpha}_s) \pi(\Bar{\Sigma}) \pi(\bm{\alpha}_s) \nonumber \\
    &\propto 2T_{nM+1}\Bigg(\bigg(\frac{\nu+p}{\nu+(\bm{\beta}-\bm{\mu})^\top\Omega^{-1}(\bm{\beta}-\bm{\mu})}\bigg)^{1/2}D_*\bm{\beta}; \Sigma_*,\nu+p \Bigg) \pi(\bm{\alpha}_s).
\end{align}
We evaluate the multivariate Student's $t$ probability on the right-hand side of (\ref{covposterior}) using the R library \texttt{tlrmvnmvt}, which implements the classic Genz algorithm \citep{Genz1999,Genz2002} and exploits a tile-low-rank algorithm \citep{Cao2020} to speed up the computation of the multivariate normal and $t$ probabilities.
To avoid sampling the correlation matrix from a constrained space, we consider the reparametrization adopted in \cite{Smith2013} and \cite{Chin2020}, which re-expresses a correlation matrix in terms of the Cholesky factor of a positive definite matrix$\bar\Sigma = \Lambda_{\bar\Sigma}^{-1/2} L_{\bar\Sigma} L_{\bar\Sigma}^\top \Lambda_{\bar\Sigma}^{-1/2}$, where $L_{\bar\Sigma}$ is a lower triangular matrix and $\Lambda_{\bar\Sigma}={\rm diag}(L_{\bar\Sigma} L_{\bar\Sigma}^\top)$.
Here the diagonal entries of $L_{\bar\Sigma}$ are set to $1$ such that the correspondence between $L_{\bar\Sigma}$ and $\bar\Sigma$ is one-to-one.
We denote the collection of the $M(M-1)/2$ unconstrained parameters in $L_{\bar\Sigma}=(l_{ij})$ by $\bm{\theta}$, i.e., $\bm{\theta}=\{l_{ij}: i>j, i,j=1,\dots,M\}$, and the $M(M-1)/2$ constrained parameters in $\bar\Sigma$ by ${\rm vec}(\bar\Sigma)$, then using a change of variables we get the posterior of $(\bm{\theta},\bm{\alpha}_s)$ as
\[
p(\bm{\theta},\bm{\alpha}_s \mid \bm{y},\bm{\beta},\nu) = p(\Bar{\Sigma},\bm{\alpha}_s \mid \bm{y},\bm{\beta},\nu) |J| = p(\Bar{\Sigma},\bm{\alpha}_s \mid \bm{y},\bm{\beta},\nu) \prod_{i=1}^M \Big(1+\sum_{j<i}l_{ij}^2\Big)^{-(M+1)/2},
\]
where $|J|=|\partial {\rm vec}(\bar\Sigma)/\partial \bm{\theta}|$ is the determinant of the Jacobian matrix of this transformation.
%Note that now the prior of $\bm{\theta}$ is
%\[
%\pi(\bm{\theta}) = \pi(\bar\Sigma)\times |J| \propto |J| = \prod_{i=1}^M \Big(1+\sum_{j<i}l_{ij}^2\Big)^{-(M+1)/2}.
%\]

As direct sampling from the distribution of $\bm{\theta},\bm{\alpha}_s \mid \bm{y},\bm{\beta},\nu$ is unknown, we propose to use a random walk Metropolis-Hastings algorithm to generate samples from it.
Specifically, we first sample $\bm{\alpha}_s^\prime$ from a proposal distribution with density $q(\cdot \mid \bm{\alpha}_s)$ and $\bm{\theta}^\prime$ from a proposal distribution with density $r(\cdot \mid \bm{\theta})$. 
%Then, the new value $(\bm{\alpha}_s^\prime, \bm{\theta}^\prime)$ is accepted with probability
%\[
%\alpha((\bm{\alpha}_s, \bm{\theta}), (\bm{\alpha}_s^\prime, \bm{\theta}^\prime)) = \min\Big\{\frac{p(\bm{\theta}^\prime,\bm{\alpha}_s^\prime \mid \bm{y}, X, \bm{\beta})\mathrm{1}((\bm{\theta}^\prime, \bm{\alpha}_s^\prime)\in C) q(\bm{\alpha}_s \mid \bm{\alpha}_s^\prime) r(\bm{\theta} \mid \bm{\theta}^\prime)}
%{p(\bm{\theta},\bm{\alpha}_s \mid \bm{y}, X, \bm{\beta})\mathrm{1}((\bm{\theta},\bm{\alpha}_s)\in C) q(\bm{\alpha}_s^\prime \mid \bm{\alpha}_s) r(\bm{\theta}^\prime \mid \bm{\theta})}, 1\Big\},
%\]
%where $\mathrm{1}(\cdot)$ is the indicator function and $C$ is the space of all $(\bm{\theta},\bm{\alpha}_s)$ such that the resulting matrix $\bar\Sigma - \bm{\delta}\bm{\delta}^\top$ is positive definite with $\bm{\delta}=(1+\bm{\alpha}^\top\Bar{\Sigma}\bm{\alpha})^{-1/2}\Bar{\Sigma}\bm{\alpha}$.
Here we take both proposal densities $q$ and $r$ as symmetric normal densities, i.e., $\bm{\alpha}_s^\prime \mid \bm{\alpha}_s \sim \mathcal{N}_M(\bm{\alpha}_s, h_1 \mathrm{I}_M)$ and $\bm{\theta}^\prime \mid \bm{\theta} \sim \mathcal{N}_J(\bm{\theta}, h_2 \mathrm{I}_J), J = M(M-1)/2$.
Then the acceptance probability is
\[
\alpha((\bm{\alpha}_s, \bm{\theta}), (\bm{\alpha}_s^\prime, \bm{\theta}^\prime)) = \min\Big\{\frac{p(\bm{\theta}^\prime,\bm{\alpha}_s^\prime \mid \bm{y}, X, \bm{\beta})\mathrm{1}((\bm{\theta}^\prime, \bm{\alpha}_s^\prime)\in C) }
{p(\bm{\theta},\bm{\alpha}_s \mid \bm{y}, X, \bm{\beta})\mathrm{1}((\bm{\theta},\bm{\alpha}_s)\in C)}, 1\Big\},
\]
where $\mathrm{1}(\cdot)$ is the indicator function and $C$ is the space of all $(\bm{\theta},\bm{\alpha}_s)$ such that the resulting matrix $\bar\Sigma - \bm{\delta}\bm{\delta}^\top$ is positive definite with $\bm{\delta}=(1+\bm{\alpha}^\top\Bar{\Sigma}\bm{\alpha})^{-1/2}\Bar{\Sigma}\bm{\alpha}$.

\subsection{MCMC Sampling Scheme}
As sampling from the distribution of $(\bm{\beta} \mid \bm{y}, X,\Bar{\Sigma},\bm{\alpha})$ is feasible using (\ref{SUTstochstic}) and sampling from the distribution of $(\Bar{\Sigma},\bm{\alpha} \mid \bm{y}, X, \bm{\beta})$ has been described in Section \ref{covsampling}, we now combine them to construct an MCMC sampler for the multivariate skew-$t$ link model.

\vspace{.2cm}
\begin{algorithm}[H]
% \KwData{this text}
% \KwResult{how to write algorithm with \LaTeX2e }
 Initialization: Set $\bm{\beta}^{(0)}, \Bar{\Sigma}^{(0)},\bm{\alpha}^{(0)}$ \;
 \For{iteration k from 1 to K}{
  [1] Sample $\bm{U}_0^{(k)}$ from $\mathcal{T}_p\big(\bm{0},\Omega-\Omega D_*^\top (D_*\Omega D_*^\top + \Sigma_*)^{-1} D_* \Omega,\nu+nM+1\big)$ (in R use \textit{rmvt})\; 
  [2] Sample $\bm{U}_1^{(k)}$ from a $(nM+1)$-variate truncated $t$ distribution with location parameter vector $\bm{0}$, dispersion matrix $s^{-1}(D_*\Omega D_*^\top + \Sigma_*)s^{-1}$, degrees of freedom $\nu$, and truncated below the level $-s^{-1}D_*\bm{\mu}$, using the accept-reject algorithm of \cite{Botev2017} (in R use \textit{mvrandt})\;
  [3] Compute $\bm{\beta}^{(k)}$ via $\bm{\beta}^{(k)}=\bm{\mu} + \big(\frac{\nu+(\bm{U}_1^{(k)})^\top s(D_*\Omega D_*^\top + \Sigma_*)^{-1}s\bm{U}_1^{(k)}}{\nu+nM+1}\big)^{1/2}\bm{U}_0^{(k)} + \Omega D_*^\top (D_*\Omega D_*^\top + \Sigma_*)^{-1} s \bm{U}_1^{(k)}$\;
  [4] Use the Metropolis-Hastings algorithm described in Section \ref{covsampling} to sample $(\bm{\theta}^{(k)},\bm{\alpha}_s^{(k)})$ from the distribution of $(\bm{\theta},\bm{\alpha}_s \mid \bm{y}, X, \bm{\beta}^{(k)})$, then return the resulting $\bar\Sigma^{(k)}$ and $\bm{\alpha}^{(k)}$.
 }
 Output: $(\bm{\beta}^{(1)},\Bar{\Sigma}^{(1)},\bm{\alpha}^{(1)}),\dots,(\bm{\beta}^{(K)},\Bar{\Sigma}^{(K)},\bm{\alpha}^{(K)})$ 
 \caption{MCMC sampling scheme for the multivariate ST link model}
 \label{algorithm1}
\end{algorithm}

\subsection{Simulation Study}
In this section, we conduct a simulation study to assess the performance of our proposed Algorithm \ref{algorithm1}.
We consider three different scenarios with different values for the degrees of freedom, i.e., $\nu=5,10,20$.
In each of the scenarios, we generate a dataset with sample size $n=50$, regression coefficients $\bm{\beta}=(\beta_1,\beta_2,\beta_3)^\top=(-1,0.5,-0.5)^\top$, skewness parameter $\bm{\alpha}_s=(\alpha_1,\alpha_2,\alpha_3)^\top=(2,0,-2)^\top$, and dispersion matrix
\[
% \bar\Sigma = \begin{pmatrix}
% 1 & 0.5 & 0 \\
% 0.5 & 1 & -0.5 \\
% 0 & -0.5 & 1 \\
% \end{pmatrix}.
\bar\Sigma = \left(
\begin{array}{ccc}
 1 & 0.5 & 0 \\
 0.5 & 1 & -0.5 \\
 0 & -0.5 & 1
\end{array}
\right).
\]
The first column of the data matrix $X$ is set to $\bm{1}$ to account for the intercept and the remaining entries in $X$ are generated from a standard normal distribution.
The intercept $\beta_1$ is chosen as $-1$ so as to obtain a highly-imbalanced dataset $\bm{y}$ with more than $80\%$ of the observations being equal to $0$.

For each of the scenarios, we fix $\nu$ to its true value, and then run Algorithm \ref{algorithm1} for 10000 iterations, discarding the first 3000 iterations as burn-in.
The prior for $\bm{\beta}$ is taken as $\mathcal{N}_p(0,25\mathrm{I}_p)$, the prior for $\bm{\alpha}_s$ is $\mathcal{N}_M(0,16)$, and the variances $h_1, h_2$ of the proposal normal densities in the Metropolis-Hastings algorithm described in Section \ref{covsampling} are taken as $h_1= h_2=0.09$.
Table \ref{tab:simu_estimates} displays the posterior estimates for each of scenarios.
The results show that the regression coefficients $\bm{\beta}$ can be quite well estimated in all scenarios, but the skewness parameter $\bm{\alpha}_s$ and dispersion matrix $\bar\Sigma$ are not easy to estimate.
This is expected as both $\bm{\alpha}_s$ and $\bar\Sigma$ determine the dependence between the different binary observations in $\bm{y}$ and such dependence enforced at the latent level cannot be easily estimated with a relatively small sample size $n$.
Moreover, the credible intervals appear to be more narrow for larger degrees of freedom $\nu$, which is due to the weaker dependence among observations (hence, larger effective sample size) implied by larger $\nu$.

\begin{table}
\centering
\caption{Posterior estimates for different scenarios in the simulation study}
\resizebox{\textwidth}{!}{\begin{tabular}{lcccccccccc}\toprule
\multicolumn{2}{c}{Scenario} &\multicolumn{3}{c}{$\nu=5$} &\multicolumn{3}{c}{$\nu=10$} &\multicolumn{3}{c}{$\nu=20$} \\
\cmidrule(r){3-5}
\cmidrule(r){6-8}
\cmidrule(r){9-11}
   & True & Est & Sd & $95\%$ CI & Est & Sd & $95\%$ CI & Est & Sd & $95\%$ CI \\\midrule
$\beta_1$ & -1 & -1.54 & 0.69 & (-3.26, -0.75) & -1.33 & 0.39 & (-2.27, -0.77) & -1.24 & 0.27 & (-1.86, -0.80) \\
$\beta_2$ & 0.5 & 0.78 & 0.39 & (0.31, 1.75) & 0.55 & 0.22 & (0.21, 1.06) & 0.52 & 0.18 & (0.21, 0.92) \\
$\beta_3$ & -0.5 & -0.46 & 0.28 & (-1.12, -0.10) & -0.76 & 0.27 & (-1.41, -0.35) & -0.71 & 0.21 & (-1.19, -0.37) \\
$\Bar{\Sigma}_{12}$ & 0.5 & 0.33 & 0.31 & (-0.35, 0.81) & 0.42 & 0.28 & (-0.18, 0.86) & 0.69 & 0.17 & (0.28, 0.94) \\
$\Bar{\Sigma}_{13}$ & 0 & -0.20 & 0.33 & (-0.79, 0.46) & -0.26 & 0.33 & (-0.81, 0.40) & -0.57 & 0.26 & (-0.93, 0.02) \\
$\Bar{\Sigma}_{23}$ & -0.5 & -0.38 & 0.32 & (-0.91, 0.29) & -0.54 & 0.31 & (-0.95, 0.21) & -0.52 & 0.30 & (-0.93, 0.20) \\
$\alpha_1$ & 2 & 0.50 & 3.62 & (-5.69, 7.46) & -1.75 & 3.24 & (-6.90, 5.35) & 0.66 & 3.69 & (-8.11, 6.66) \\
$\alpha_2$ & 0 & -2.89 & 3.33 & (-9.81, 2.74) & 0.39 & 3.18 & (-4.78, 8.31) & 2.88 & 4.16 & (-5.15, 11.32) \\
$\alpha_3$ & -2 & 2.62 & 3.72 & (-5.19, 9.20) & -0.49 & 4.59 & (-7.43, 9.79) & -0.81 & 3.02 & (-6.46, 5.53) \\
\bottomrule
\end{tabular}}
\label{tab:simu_estimates}
\end{table}

\subsection{Application to COVID-19 Pandemic Data}\label{application}
In this section we illustrate our methodology to COVID-19 pandemic data from different counties of the state of California, USA, freely downloaded from the California open data portal https://data.ca.gov.
The dataset contains the number of daily new confirmed cases and deaths from March 18, 2020, to November 24, 2020, in 58 counties of California.
There is a clear weekly cyclic pattern in this dataset, i.e., the numbers of new confirmed cases on weekdays are often much larger than those during the weekends.
This is possibly due to the fact that people tend to enjoy their weekends and go to the hospital for testing after the weekend, or some testing facilities are closed during weekends.
To avoid modeling this artificial cyclic pattern, we aggregate the data and consider the weekly new confirmed cases, resulting in $n=36$ weekly observations.
As $n$ is relatively small, we here only focus on the three most populous counties in California, i.e., Los Angeles, San Diego and Orange.
Our goal here is to jointly model the occurrence of extreme spikes in new weekly cases, i.e., ``abnormal'' weeks with respect to the overall expected trend, and to detect if the spikes are spatially correlated.
This informs us about the potential transmission modes of the virus between counties, and whether an outburst in one county may lead to an increased number of cases in another county.

To remove the obvious trend, we apply smoothing splines with five knots to the logarithm of each of the three time series, where the logarithm is used because most epidemics grow approximately exponentially during the initial phase \citep{Ma2020}.
Alternatively, one can try to fit the well-known susceptible-infected-recovered (SIR) model \citep{SIR1,SIR2} to remove the trend.
Unfortunately, this is not feasible with our dataset as it does not contain the number of daily recovered cases.
Figure \ref{smoothing} displays the observed data for the three counties, the smoothing splines for each time series and the resulting residuals.
We then consider a residual point as an extreme spike if it exceeds the empirical $90\%$ quantile of the corresponding time series, and we denote it as $1$; otherwise we denote it as $0$.
In this way, we get three imbalanced binary time series and we aim to model the dependence among them.
\begin{figure}
    \centering
    \includegraphics[scale=0.85]{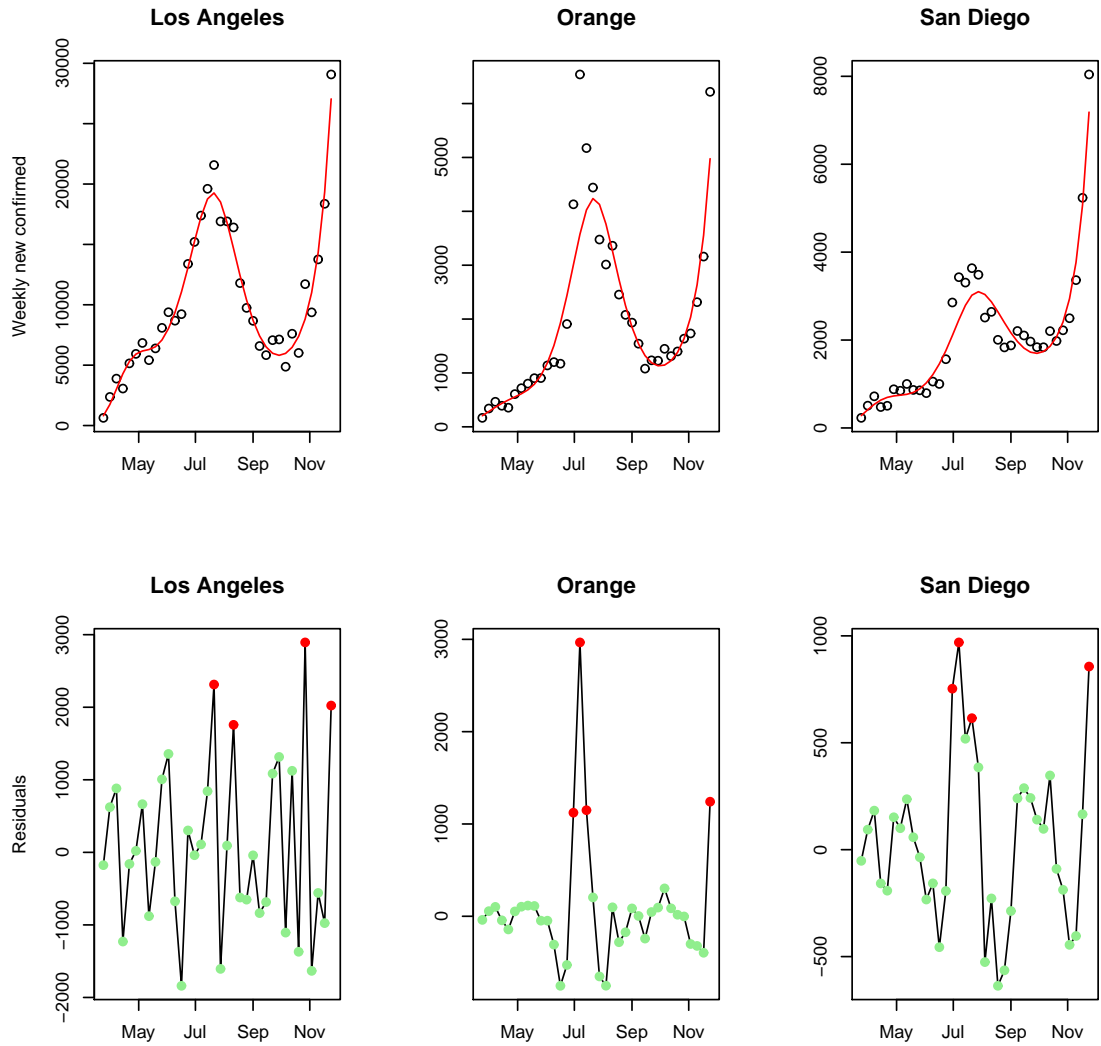}
    \caption{Upper panel: smoothing splines for the time series of weekly new confirmed cases at Los Angeles (left),  Orange (middle), and San Diego (right). Lower panel: the residuals obtained as the difference between the original data and the fitted splines, with red points considered as extreme spikes and green points as non-extreme values.}
    \label{smoothing}
\end{figure}

We consider three covariates in total, i.e., an intercept, one covariate as time, and another one as the square of time.
Following the recommendation of \cite{Gelman2008}, we standardize the two temporal predictors in a preliminary step to make them have mean $0$ and standard deviation $1$.
To assess the performance of the multivariate skew-normal link model, we consider  six models $\mathcal{M}_1,\mathcal{M}_2,\mathcal{M}_3,\mathcal{M}_4, \mathcal{M}_5, \mathcal{M}_6$ of different complexity.
$\mathcal{M}_1,\mathcal{M}_2,\mathcal{M}_3$ are the multivariate skew-$t$ link model with $\nu=5,10,20$, respectively, $\mathcal{M}_4$ is the multivariate skew-normal model (i.e., obtained as $\nu\to\infty$), $\mathcal{M}_5$ is the multivariate probit model (obtained with $\nu\to\infty$ and $\alpha=0$), and $\mathcal{M}_6$ is the independent probit model (obtained with $\nu\to\infty$, $\bar\Sigma=\mathrm{I}, \alpha=0$).

For each of these models, we run the Algorithm \ref{algorithm1} for 25000 iterations and remove the first 5000 samples as burn-in.
The prior for the regression parameters $\bm{\beta}$ is specified as $\mathcal{N}_p(\bm{0},25\mathrm{I}_p)$, and the prior for the skewness parameters is taken as $\mathcal{N}_M(\bm{0},16\mathrm{I}_M)$.
The variances of the proposal densities in the Metropolis-Hastings algorithm are taken as $h_1=h_2=0.09$.

Table \ref{tab:model_estimates} summarizes the estimation results for all the models.
The results show that the estimate of the intercept for all the models are almost the same and are significantly negative.
This is expected as $90\%$ of the observations are $0$ and only $10\%$ are $1$.
We also observe that the confidence intervals for the correlation and skewness parameters are generally quite large (as in the simulation study), implying that they are hard to estimate with only $n=36$ observations.
However, the correlation between the counties of Orange and San Diego, i.e., $\Bar{\Sigma}_{23}$, seems to be quite strong, as its posterior mean for models $\mathcal{M}_1, \mathcal{M}_2, \mathcal{M}_3, \mathcal{M}_4$, and $\mathcal{M}_5$ is consistently far from zero (with an estimate close to 0.78) and its $95\%$ credible interval always excludes zero.
This indicates that these two counties are more connected together in terms of extreme COVID-19 cases than the other pairs of counties considered, which sheds some light into the spread of the epidemic.
The extreme occurrences observed in the counties of Los Angeles and San Diego also seem fairly strongly interconnected since the estimate of $\Sigma_{13}$ is also quite high, yet to a milder degree.

\begin{table}
\centering
\caption{Posterior estimates for different models fitted in our COVID-19 data application in Section \ref{application}}
\resizebox{\textwidth}{!}{\begin{tabular}{lccccccccc}\toprule
&\multicolumn{3}{c}{$\mathcal{M}_1$} &\multicolumn{3}{c}{$\mathcal{M}_2$} &\multicolumn{3}{c}{$\mathcal{M}_3$} \\
\cmidrule(r){2-4}
\cmidrule(r){5-7}
\cmidrule(r){8-10}
   & Est & Sd & $95\%$ CI & Est & Sd & $95\%$ CI & Est & Sd & $95\%$ CI \\\midrule
Intercept & -1.67 & 0.77 & (-3.62, -0.76) & -1.58 & 0.49 & (-2.76, -0.86) & -1.48 & 0.37 & (-2.34, -0.89) \\
Time & 1.72 & 1.68 & (-0.67, 5.80) & 1.62 & 1.36 & (-0.54, 4.75) & 1.54 & 1.27 & (-0.60, 4.43) \\
Time$^2$ & -1.19 & 1.38 & (-4.47, 0.99) & -1.13 & 1.15 & (-3.73, 0.81) & -1.07 & 1.08 & (-3.45, 0.82) \\
$\Bar{\Sigma}_{12}$ & 0.24 & 0.29 & (-0.36,0.74) & 0.23 & 0.46 & (-0.33, 0.72) & 0.24 & 0.28 & (-0.34, 0.73) \\
$\Bar{\Sigma}_{13}$ & 0.54 & 0.24 & (-0.01, 0.89) & 0.51 & 0.23 & (0.02, 0.87) & 0.51 & 0.24 & (-0.03, 0.88) \\
$\Bar{\Sigma}_{23}$ & 0.78 & 0.16 & (0.40, 0.97) & 0.76 & 0.15 & (0.40, 0.97) & 0.73 & 0.17 & (0.29 0.95) \\
$\alpha_1$ & 1.86 & 4.69 & (-6.81, 10.71) & -2.04 & 3.13 & (-7.31, 4.47) & -0.26 & 2.94 & (-5.64, 5.54) \\
$\alpha_2$ & -0.82 & 3.49 & (-6.74, 6.66) & -0.45 & 3.13 & (-6.00, 4.89) & 1.05 & 3.41 & (-5.64, 6.82) \\
$\alpha_3$ & -0.42 & 3.02 & (-7.13, 4.91) & 0.05 & 4.37 & (-7.02, 8.75) & -0.20 & 4.83 & (-10.34, 7.70) \\
\midrule
 &\multicolumn{3}{c}{$\mathcal{M}_4$} &\multicolumn{3}{c}{$\mathcal{M}_5$} &\multicolumn{3}{c}{$\mathcal{M}_6$} \\
\cmidrule(r){2-4}
\cmidrule(r){5-7}
\cmidrule(r){8-10}
   & Est & Sd & $95\%$ CI & Est & Sd & $95\%$ CI & Est & Sd & $95\%$ CI \\\midrule
Intercept & -1.40 & 0.27 & (-1.98, -0.92) & -1.42 & 0.26 & (-1.97, -0.96) & -1.36 & 0.20 & (-1.77, -1.00) \\
Time & 1.47 & 1.18 & (-0.62, 4.05) & 1.47 & 1.18 & (-0.56, 4.05) & 1.23 & 0.89 & (-0.37, 3.13) \\
Time$^2$ & -1.01 & 1.02 & (-3.13, 0.86) & -1.02 & 1.02 & (-3.17, 0.79) & -0.84 & 0.77 & (-2.45, 0.59) \\
$\Bar{\Sigma}_{12}$ & 0.27 & 0.28 & (-0.32, 0.76) & 0.25 & 0.30 & (-0.34, 0.77) & \\
$\Bar{\Sigma}_{13}$ & 0.62 & 0.24 & (0.01, 0.93) & 0.53 & 0.24 & (0.01, 0.90) &  \\
$\Bar{\Sigma}_{23}$ & 0.78 & 0.16 & (0.37, 0.98) & 0.74 & 0.16 & (0.34, 0.95) & \\
$\alpha_1$ & 1.65 & 5.16 & (-5.29, 11.63) &  \\
$\alpha_2$ & -0.39 & 3.08 & (-6.83, 6.43) & \\
$\alpha_3$ & 0.39 & 2.54 & (-4.54, 5.53) &  \\
\bottomrule
\end{tabular}}
\label{tab:model_estimates}
\end{table}

To compare the different fitted models, we use the Deviance Information Criterion (DIC) proposed by \cite{DICmeasure}.
The DIC is the Bayesian analogue of the Akaike Information Criterion (AIC) and is defined as
\[
\text{DIC} = D(\Bar{\bm{\tau}}) + 2p_D,
\]
where $\bm{\tau}$ denotes the collection of all the parameters, $\Bar{\bm{\tau}} = \mathrm{E}[\bm{\tau}\mid \bm{y}]$ is its posterior mean, $D(\cdot)$ is a deviance function and $p_D=\mathrm{E}[D(\bm{\tau})\mid \bm{y}]-D(\Bar{\bm{\tau}})$ is the effective number of model parameters.
Here we take the deviance function $D(\bm{\tau})$ as $-2\log p(\bm{y} \mid \bm{\beta},\Bar{\Sigma},\bm{\alpha},\nu)$ when the model is the skew-$t$ link model, or $-2\log p(\bm{y} \mid \bm{\beta},\Bar{\Sigma},\bm{\alpha})$ when the model is the skew-normal link model, and estimate $\mathrm{E}[D(\bm{\tau})\mid \bm{y}]$ by Monte Carlo using the samples generated from Algorithm \ref{algorithm1}.
The smaller the DIC value, the better the model's goodness-of-fit and predictive performance.
We refer to \cite{DICmeasure} for other properties about the DIC measure.

Table \ref{tab:model_DIC} reports the estimated DIC values for the six different models.
The results show that the multivariate skew-normal model $\mathcal{M}_4$ provides the best fit to the data despite its high complexity, the multivariate probit model $\mathcal{M}_5$ is the second best, and the independent symmetric probit model $\mathcal{M}_6$ is the worst.
This has two major implications.
The first is that spatial dependence plays an important role in the spread of the epidemic and ignoring the correlation would lead to a poor fit of the extreme spikes.
The second is that adding the skewness parameter indeed improves the model's flexibility and can provide a better fit to our highly imbalanced dataset.
\begin{table}
\centering
\caption{Estimated DIC values for the different models fitted in our COVID-19 data application in Section \ref{application}}
\begin{tabular}{ccc}\toprule
Model & $\#$ of parameters  &  DIC \\ \midrule
$\mathcal{M}_1$ & 9 & 67.93  \\
$\mathcal{M}_2$ & 9 & 68.06  \\
$\mathcal{M}_3$ & 9 & 67.84  \\
$\mathcal{M}_4$ & 9 & $\bm{65.77}$  \\
$\mathcal{M}_5$ & 6 & 67.12  \\
$\mathcal{M}_6$ & 3 & 78.97  \\
\bottomrule
\end{tabular}
\label{tab:model_DIC}
\end{table}

% Conclusion
\section{Conclusion}
\iffalse
In this article we have proposed a multivariate skew-elliptical link model for correlated binary responses and proved that its regression coefficients have closed-form unified skew-elliptical posterior.
We have shown how to perform Bayesian inference for its special case, the skew-$t$ link model, and the Bayesian inference for the skew-normal link model can be done similarly.
The new methodology was illustrated with COVID-19 data from three different counties of the state of California, USA.
The results showed that the spatial dependence plays an important role in the spread of the epidemic and cannot be neglected when modeling the occurrence of extreme spikes in weekly new confirmed cases.
The results also demonstrated that adding the skewness parameters indeed improves the flexibility of the multivariate probit model and can provide a better fit, especially for highly imbalanced data.
\fi

Although we here focus on the skew-elliptical link model, the result of a closed-form posterior for the regression coefficients could also be obtained if we consider a more flexible class of distributions for the assumption (\ref{modelassumption}).
In fact, if $\bm{\varepsilon}\mid \bm{\beta},\Sigma,\bm{\alpha}$ has a distribution which is closed under affine transformation, following the proof of Lemma \ref{lemma1} and Theorem \ref{theorem1}, one can show that the posterior of $\bm{\beta}$ coincides with a fundamental skew distribution \citep{AreGenton2005}.
This novel result opens up new avenues for the development of skewed link models for correlated binary data.

There are various directions for future research.
As the number of observations in our dataset is relatively small, we chose not to consider too many covariates and restricted the number of counties.
An interesting extension of our real data application would be to consider a larger dataset with more informative covariates, such as daily weather information or population migration between different counties. 
Adding such extra covariates could potentially fit the data better and provide a more detailed and informed explanation of the spread of epidemic.
Another interesting methodological extension is to improve Algorithm \ref{algorithm1}.
As we used the accept-reject algorithm of \cite{Botev2017} within Algorithm \ref{algorithm1} to sample from a multivariate truncated $t$ distribution, its lack of scalability to higher dimensions is inevitably inherited.
Therefore, more efficient algorithms to sample from high-dimensional truncated normal and $t$ distributions would significantly improve the speed of algorithm \ref{algorithm1}.
Finally, in the simulation study and data application we chose to fix the degrees of freedom $\nu$ to three different values to facilitate inference.
If one has many more observations and a more efficient algorithm to sample from high-dimensional truncated $t$ distribution, one can alternatively include the estimation of $\nu$ in the Metropolis-Hastings algorithm described in Section \ref{covsampling}.

% References

\small{
\bibliography{reference}
}
\end{document}